\documentclass[12pt,reqno,a4paper]{amsart}
\usepackage[margin=3cm,footskip=1cm]{geometry}

\usepackage{amssymb} \usepackage{amsmath} \usepackage{amsfonts}
\usepackage{amsthm} \usepackage{mathtools}
\usepackage{esint}

\usepackage{mathabx}
\usepackage{enumerate} \usepackage[latin1]{inputenc}
\usepackage[shortlabels]{enumitem}

\usepackage{color}
\usepackage{graphicx} \usepackage{verbatim}

\newcommand{\supp}{\operatorname{supp}}

\newcommand{\tr}{{\operatorname{tr}}}

\newcommand{\Ran}{{\operatorname{Ran}}}

\newcommand{\R}{{\mathbb{R}}} \newcommand{\Z}{{\mathbb{Z}}}
\newcommand{\C}{{\mathbb{C}}}

\newcommand{\e}{{\rm e}}

\renewcommand{\tr}{{\rm tr}}

\theoremstyle{plain}%default
\newtheorem{thm}{Theorem}[section]
\newtheorem{proposition}[thm]{Proposition}
\newtheorem{lemma}[thm]{Lemma} \newtheorem{corollary}[thm]{Corollary}
\theoremstyle{definition}

 \newtheorem{remark}[thm]{Remark}
\newtheorem{remarks}[thm]{Remarks}

\newtheorem*{remarks*}{Remarks}
\newtheorem*{remark*}{Remark}
\newtheorem*{defn*}{Definition}
\newtheorem{conjecture}[thm]{Conjecture}

\numberwithin{equation}{section}

\title{The Howland - Kato Commutator Problem}

\author{Ira Herbst, Thomas  L. Kriete}
\address{Department of Mathematics \\
  University of Virginia \\
  Charlottesville \\
  VA 22904\\ U.S.A.}
\email{iwh@virginia.edu, tlk8q@virginia.edu}

\date{\today}

\begin{document}

\maketitle

\begin{abstract}
We investigate the following problem: For what $f$ and $g$ is the commutator $i[f(P),g(Q)]$  positive when $f$ and $g$ are bounded measurable functions?  This problem originated in work of James Howland and was pursued by Tosio Kato who suggested what might be the answer.  So far there is no proof that Kato was correct but in this paper we discuss the problem and give some partial answers to the above question.
\end{abstract}

\tableofcontents

\section{Introduction}

In a paper on spectral theory \cite{JSH}, J. Howland used the positive commutator of two bounded functions of the Heisenberg operators $P$ and $Q$,  $ i[f(P), g(Q)]$, as a technical tool.  Here $ P = -id/dx$ and $Q$ is multiplication by $x$ in $L^2(\R)$.  The functions $f$ and $g$ were specifically
\begin{equation}
f(t) = \tan^{-1}(t/2), g(t) = \tanh(t).
\end{equation}
He sent his paper to T. Kato, his former thesis advisor, who got interested in the more general question: for what bounded real functions is the above commutator positive?  Kato made much progress on this problem and in a beautiful paper, \cite{TK}, he identified a very interesting class of pairs of such functions for which the commutator was positive.  In this paper we will always assume that $f$ and $g$ are bounded measurable real functions.
To state the main result of \cite{TK}, for $a >0$ define 

 $$K_a = \{f:\R \to \R| \ f \text{\ is bounded and has an analytic continuation to the strip}$$
 $$\ |\text{Im}z| < a \text{ with } \text{Im}f(z) \text{Im}z \ge 0 \}.$$
 Then 
 
 \begin{thm}[Kato]\label{K's thm}
 If $f \in K_a$, $g \in K_b$ with $ab \ge \pi/2$, then $ i[f(P), g(Q)] \ge 0$.  
 \end{thm}

As we will see later, if $f$ and $g$ are two bounded real measurable functions for which $ i[f(P), g(Q)] = C \ge 0 $, then as noted by Kato, $C$ is trace class.  It is thus natural to look at the case where $C$ is a non-zero, rank one operator.  Kato does this in  \cite{TK} and shows that in this case there exist $a$ and $b$ with $ab = \pi/2$ such that $\pm f \in K_a$ and $\pm g \in K_b$ (the signs are correlated).  In fact $f(x) = c_1 + d_1 \tanh(\hat{a} (x - t_1)) $ and $g(x) = c_2 + d_2\tanh(\hat{b} (x - t_2)) $ where $c_j$, $d_j$ and $t_j$ are real, $d_1 d_2 > 0$, and (following Kato) $\hat{a} = \pi/2a$ and $\hat{b} = \pi/2b$, so that $\hat{a} = b$ and $\hat{b} = a$.  (Kato assumes that $f$ and $g$ are absolutely continuous with derivatives in $L^1(\R)$ but this is not necessary as we will see later.)  It is clear that from these functions, more pairs of functions with positive commutators can be constructed by convolution with a positive measure.  In fact as Kato shows, the family of $f$ of the form   

\begin{equation} \label{integral rep}
f(x) = \int_{\R} \tanh \hat{a} (x -t) d\nu(t) + c
\end{equation}
exhausts all of $K_a$, if $c$ is real and $\nu$ is a finite positive measure. 
These results led Kato to state (in \cite{TK}) \\

``In fact there is some reason to believe that these [$f\in K_a, g\in K_b$ with $ab\ge \pi/2$] are the only solutions to [ $ i[f(P), g(Q)] \ge 0 $]." \\

Note that since $K_a \supset K_c$ whenever $0 < a < c$, the meaning of Kato's statement just quoted, as well as the meaning of the statement of Theorem 1.1, are unchanged if the inequality $ab \ge \pi/2$ is replaced by the equality $ab = \pi/2$.\\

Kato also shows that if   $ i[f(P), g(Q)] =0$ and both $f$ and $g$ are absolutely continuous with $L^1(\R)$ derivatives then at least one of them is constant.  In much of this paper we will relax this assumption and only assume that $f$ and $g$ are bounded measurable real functions. This opens up another interesting possibility.  We will show that for $f$ and $g$ bounded, real, and measurable, 

\begin{thm} \label{commutes}
The commutator $ [f(P), g(Q)] =0$ if and only if either $f$ or $g$ is almost everywhere constant or both have periodic versions with periods $\tau_f$ and $\tau_g$ satisfying $\tau_f\tau_g = 2\pi$.
\end{thm}
Actually the theorem is still true without the assumption that $f$ and $g$ are real.\\

There is a striking difference in the set of allowed $f$ and $g$ when  $C = i[f(P),g(Q)] \ge 0 $ and in addition we impose $C \ne 0$.  We will show

\begin{thm} \label{monotone}
Suppose $C= i[f(P),g(Q)] \ge 0 $ and in addition $C \ne 0$.  Then there are versions of $f$ and $g$ which are both monotone (either both increasing or both decreasing).
\end{thm}
For convenience let us formulate a conjecture which we shall call K (after Kato):
\begin{conjecture}
Suppose $C= i[f(P),g(Q)] \ge 0 $ and in addition $C \ne 0$.  Suppose both $f$ and $g$ are increasing.  Then there exist $a$ and $b$ with $ab=\pi/2$ such that $f\in K_a$ and $g\in K_b$.
\end{conjecture}

Unfortunately we are far from proving K.  But in the rest of this paper we will give several results which illuminate the properties of the set of $f$ and $g$ for which the commutator is positive.  In addition to the theorems of this Introduction and their proofs, see in particular the section on $2 \times 2$ positivity which gives inequalities which $f$ and $g$ must satisfy under a mild assumption.  

\begin{thm}
Suppose $C= i[f(P),g(Q)] \ge 0 $ and in addition $C \ne 0$.    Suppose both $f$ and $g$ are increasing.  Then both $f$ and $g$ have continuous versions which are  strictly increasing.  Taking $f$ and $g$ to be continuous, their inverse functions are absolutely continuous.
\end{thm}

A hint that K might be true is the following result:

\begin{thm} \label{half}
Suppose $g$ is non-constant, lies in $K_b$, and has the integral representation 
\begin{equation}\label{tom}
g(x) = \int_{\R} \tanh \hat{b} (x -t) d\mu(t) + d
\end{equation}
where $d$ is real and $\mu$ is a finite positive measure such that $$\int |t|e^{2\hat{b}t} d\mu(t)< \infty$$  or $$\int |t|e^{-2\hat{b}t} d\mu(t)< \infty.$$  If $f \in L^{\infty}(\R)$ and $i[f(P),g(Q)] \ge 0 $ then $f \in K_{\hat b}$.
\end{thm} 
There is a reason that a condition such as this assumption of exponential decay is required.  Just because we have assumed $g \in K_b$ does not mean that $g$ is not in a smaller class  $K_{e}$  with
 $e>b$.  Suppose $c$ is the largest number $e$ such that $g \in K_e$.  Then $\hat{c} = \pi/2 c < \hat{b}$.  We could not be able to prove that $f \in K_{\hat{b}}$ when only $f \in K_{\hat{c}}$ is required.  To see that our exponential decay assumption eliminates this possibility, note that if $b$ and $c$ are as above, then besides (\ref{tom}) there is another representation of $g$ of the form
 $$ g(x ) =  \int_{\R} \tanh \hat {c}(x-t) d\mu_1(t) + d_1$$ where $\mu_1$ is a finite positive measure and $d_1$ is a real constant.  We have an explicit formula for the imaginary part of $g$ on the line $z = x + ib$,
  $$\text{Im} g(x+ ib)= \int_{\R} \sin( b\pi/c)(\cosh(2\hat{c}(x-t)) + \cos(b\pi/c))^{-1}d\mu_1(t) \ge k (\cosh 2\hat{c}x)^{-1}$$ for some positive $k$.
 If we note that (see \cite{TK}) $d\mu(t) = (2\pi)^{-1}\text{Im} g(t+ ib)dt$, it follows that $ \int |t|e^{\pm 2\hat{b}t} d\mu(t) = \infty$.  \\
 
The analyticity requirements of the conjecture K remind us of a theorem of Loewner \cite{Loe}:

\begin{thm}
The real measurable function $g$ defined on the interval $(a,b)$ has the property that for any two self-adjoint operators $A$ and $B$ with spectrum in $(a,b)$ and satisfying $A\ge B$ we have $g(A) \ge g(B)$ if and only if $g$ has an analytic continuation to $\{\text{Im}z \ne 0\}\cup (a,b)$ satisfying $\text{Im}g(z)\text{Im}z \ge 0$.
\end{thm}
The function $g$ is called operator monotone.\\
We can formulate the positivity of the commutator $ i[f(P),g(Q)]$ in a way that makes a connection with Loewner's theorem in the following way, at least formally (since we have not shown that $f$ is absolutely continuous for example).  Since
$$d/dt \big(e^{itf(P)} g(Q) e^{-itf(P)}\big) = e^{itf(P)} i[f(P),g(Q)] e^{-itf(P)} ,$$
if the commutator is positive and non-zero we have 
$$g(e^{itf(P)}Qe^{-itf(P)}) = g(Q + tf'(P)) \ge g(Q)$$ for positive $t$.  As we have seen we can assume that $f$ (and $g$) are increasing so that at least formally $Q +t f'(P) \ge Q $ for positive $t$.  Thus with $A = Q + tf'(P)$ and $B=Q$ we have $A\ge B$ for $t$ positive while 
$$g(A) \ge g(B) \ \text {for all such} \ t$$ 
is the same as the positivity of the commutator. \\

In the following sections we prove theorems 1.2, 1.3, 1.5, and 1.6 and add some further information about this fascinating problem.

If $C$ is an operator and $C \ge 0$ we will say that $C$ is positive, although perhaps non-negative would be more accurate.  Similarly we sometimes call a non-decreasing function an increasing function.  For the inner product of two vectors $h$ and $k$ in Hilbert space we write $(h,k)$, linear in $k$ and conjugate-linear in $h$.  The Fourier transform of a function $h$ on $\R$ is denoted by $\hat h$ and includes the factor $(2\pi)^{-1/2}$ while $\check h$ denotes the inverse Fourier transform.  

\section{Acknowledgement}
We are grateful to Brian Hall for many useful conversations about this problem.

\section{Finite rank commutators: $[i\tanh \alpha P, \tanh \beta Q]$}\label{finiterankC}

In this section we consider the commutator $[if(P), g(Q)]$ where $f$ and $g$ are the basic functions from which all functions in the Kato classes $K_a$ are constructed.  We see that 
$$[i\tanh \alpha P, \tanh \beta Q] = 4i[(1+ e^{2\alpha P})^{-1}, (1 + e^{2\beta Q})^{-1})]$$
so that the positivity of the commutator with $\tanh$ is just a statement about the positivity of the commutator of resolvents of the exponential function.  For that reason it is interesting to note the easily verified fact that for $\alpha$ and $\beta$ real and $(2\alpha)(2 \beta) = 2n\pi$ with $n \in \Z$
$$e^{2i\alpha P}e^{2i\beta Q}e^{-2i\alpha P} = e^{2i\beta (Q +2\alpha) } = e^{2i\beta Q} $$
and thus
$$ [(\lambda + e^{2i\alpha P})^{-1}, (\lambda + e^{2i\beta Q})^{-1})] = 0.$$
for $|\lambda| \ne1$.  A formal calculation with unbounded operators would yield 
$$e^{2\alpha P}e^{2\beta Q}e^{-2\alpha P} = e^{2\beta (Q -2i\alpha) } = e^{2\beta Q} $$
thus leading to
$$ [(1+ e^{2\alpha P})^{-1}, (1 + e^{2\beta Q})^{-1})] = 0$$
if $(2\alpha)( 2\beta) = 2n\pi$.  But this is incorrect as we see from the next proposition.  

\begin{proposition}  \label{tanh commutator}
The integral kernel of the commutator $ i[\tanh \alpha P, \tanh \beta Q] =4i[(1 + e^{2\alpha P})^{-1}, (1 + e^{2\beta Q})^{-1}]$, for $\alpha$ and $\beta$ real and  $(2\alpha)(2\beta) = 2\pi n > 0$,  is given by 

\begin{equation}
(\beta/n\pi)\sum_{k = 0}^{n-1}\psi_k(x) \phi_k(y)
\end{equation}
where $$\psi_k(x) = (\cosh \beta x)^{-1}e^{(n-1 - 2k)\beta x/n}$$ and  $$\phi_k(x) = (\cosh \beta x)^{-1}e^{-(n-1 - 2k)\beta x/n}$$ 
\end{proposition}
Thus  $ i[\tanh \alpha P, \tanh \beta Q] $ is a rank $n$ operator if $(2\alpha)( 2\beta) = 2\pi n$.  Note that for the rank one case (the one considered by Kato) where 
$(2\alpha)( 2\beta) = 2\pi$, the commutator is positive.  We have assumed $n > 0$.  Reversing the sign of $n$ just amounts to a sign  change in the commutator.
If $n>1$ this operator is not positive.  This can be seen by looking at the matrix $K_{ij}: = K(x_i,x_j)$ in the $2\times 2$ case.  We have
\begin{equation}
K_{ij} = (\beta/n\pi)(\cosh \beta x_i \cosh \beta x_j)^{-1} \sum_{k=0}^{n-1}f_k(x_i-x_j)
\end{equation}
where $f_k(x) = e^{(n-(2k +1))\beta x/n}$
and thus in the $2 \times 2 $ case 
\begin{align}
&\det (K_{ij}) = K_{11}K_{22} - K_{12}^2=\\
&(\beta/n\pi)^2 (\cosh \beta x_1 \cosh\beta x_2)^{-2}(n^2 - |\sum_{k=0}^{n-1}f_k(x_1-x_2)|^2)
\end{align} 
Using the fact that for a positive number $a \ne 1$ we have $a + a^{-1} > 2$, we see that  if $x_1-x_2 \ne 0$ and $n>1$, $\sum_{k = 0}^{n-1} f_k(x_1-x_2)> n$ and thus the determinant is negative.

\begin{remark}
The non-commutativity of $e^{\alpha P}$ and $e^{\beta Q}$ for $\alpha \beta = \pm 2\pi$ has been a source of counterexamples for the uniqueness of the representation of the canonical commutation relations \cite{Fu} and the hypotheses under which the so called virial theorem is true \cite {GG}.  In the latter reference one can also see in what sense these operators commute.  
\end{remark}

Before proving Proposition \ref{tanh commutator} we need

\begin{lemma} \label{kernel}

If $f$ and $g$ are bounded, $f$ is monotone increasing, and $g$ is smooth with bounded derivatives then for $\psi \in L^2(\R) $

\begin{equation} \label{kernel}
( i[f(P), g(Q)] \psi)(x) =  \int K(x,y) \psi(y) dy  
\end{equation}

where 

\begin{equation} \label{kernel2}
K(x,y) = \frac{1}{\sqrt{2\pi}} \frac{g(x) - g(y)}{x-y} \widehat {df}(y-x).
\end{equation} \\

If $g$ is smooth with bounded derivatives, and $f$ is bounded then if $\psi \in \mathcal{S}(\R)$ (the Schwartz space) we have 

$$( i[f(P), g(Q)] \psi)(x) = \int  \frac{1}{\sqrt{2\pi}} \frac{g(x) - g(y)}{x-y} \widehat {f'}(y-x)\psi(y)dy$$
where $f$ is considered a tempered distribution.

\end{lemma}

\begin{remarks}
If $f$ is monotone increasing, we take the version of $f$ which is right continuous in the definition of $df$.
See Lemma \ref{kernel2} for a more complete result.
\end{remarks}

\begin{proof}

We first assume $\psi \in \mathcal{S}(\mathbb{R})$.  $$(i[f(P), g(Q)]\psi)(x) =i\int \left(\int f(\xi) (g(y) - g(x)) \psi(y) e^{-i\xi y } dy\right) e^{i\xi x} d\xi /2\pi.$$ Let 
$$h_x(y) = \frac{g(y+x) - g(x)}{y} \psi (y+ x)/\sqrt{2\pi}.$$  Then 
$$(i[f(P), g(Q)]\psi)(x) =i\int f(\xi) i(\widehat {h_x})'(\xi) d\xi = \int \widehat{h_x}(\xi) df(\xi) = \int h_x(y) \widehat {df}(y)dy$$
where in the second equality we have used the integration by parts formula \cite{F} in the case where we assume $f$ is monotone.
After a change of variable this gives the first result (\ref{kernel}) for $\psi \in  \mathcal{S}(\mathbb{R})$.  A limiting argument gives (\ref{kernel}) for $\psi \in L^2(\R)$.  If we only assume $f$ is bounded, the first equality in the last equation gives our result when we note that $h_x \in \mathcal{S}(\R)$ and remember the definition of the derivative and the Fourier transform of a tempered distribution.

\end{proof}
\begin{proof}[Proof of Proposition \ref{tanh commutator}] 

The kernel $K(x,y)$ is given by $$(2\pi)^{-1/2} \frac{\tanh \beta x - \tanh \beta y}{x-y} \hat f'(y-x)$$ with $f(\xi) = \tanh \alpha \xi$.  Using the known Fourier transform of $(\cosh x)^{-2}$ (see (\ref{FT})) we obtain $$K(x,y) = (\beta/n\pi)\frac{\tanh \beta x - \tanh \beta y}{\sinh( \beta(x-y)/n)}.$$  We calculate $$\tanh \beta x - \tanh \beta y = \sinh \beta (x-y) (\cosh \beta x \cosh \beta y)^{-1}.$$ Using $$a^n - b^n = (a -b) ( a^{n-1} + a^{n-2} b + \cdots + b^{n-1})$$ with $a = e^{\beta x/n}, b = e^{-\beta x/n}$ we obtain $$ \frac{\sinh \beta x }{ \sinh \beta x/n} = \sum_{k = 0}^{n-1} e^{(n-(2k +1))\beta x/n}$$ which gives the result.  
\end{proof}

\section{Theorem \ref{commutes} - the case $[f(P),g(Q)] =0$}

In this section we consider the case where for real bounded measurable functions $f$ and $g$, we have $C=i[f(P),g(Q)] = 0$.

\begin{proof}[Proof of theorem \ref{commutes}]
If $f$ is periodic with period $a = \tau_f$ then $e^{iaQ} f(P) e^{-iaQ} = f(P-a) = f(P)$. Similarly $e^{-iaQ}$ commutes with $f(P)$.  If $g$ is periodic with period $\tau_g = 2\pi/\tau_f$ then $g(Q)$ is a function of $e^{iaQ}$ and thus commutes with $f(P)$.

For the converse first assume $f$ and $g$ are $C^{\infty}$ with bounded derivatives and $f(P)$ and $g(Q)$ commute.  Suppose $f$ is not constant and suppose $k_1$ and $k_2$ are two points so that $f(k_1) \ne f(k_2)$.  Choose open neighborhoods $N_j$ of $k_j$ with $\overline {N_j}$ compact such that $f(\overline {N_1})\cap f(\overline {N_2}) = \emptyset$.  Choose a bounded continuous  function $F$ with $F= 0$ on $f(\overline {N_1})$ and $F = 1$ on $f(\overline {N_2})$.  Note that $(F\circ f)(P)$ commutes with $g(Q)$ so that  if $\widehat {\phi_j} \in C_0^{\infty}(N_j)$ 

$$0 = (\phi_1, [F\circ f(P), g(Q)]\phi_2) = -(\phi_1, g(Q) \phi_2).$$
Considering $g$ as a tempered distribution and denoting $\psi_1 = \overline {\hat{\phi_1}} \in  C_0^{\infty}(N_1), \psi_2 = \hat{\phi_2} \in  C_0^{\infty}(N_2)$ we have 

$$ 0 = g(\overline{\phi_1}\phi_2) = \hat{g}_r(\widehat{\overline{\phi_1}\phi_2})= \int \hat{g_r}(\psi_2(\xi + \cdot) \psi_1(\xi)d\xi/\sqrt{2\pi}$$ where $h_r(x) = h(-x)$.
Since this is true for all $\psi_1  \in  C_0^{\infty}(N_1), \ \hat{g}_r(\psi_2(\cdot + \xi)) = 0$ for all $\xi \in N_1$ and thus $\hat{g}_r(\psi) = 0$ for all $\psi \in  C_0^{\infty}(N_2 - \xi)$ for all $\xi \in N_1$ or more concisely  $\hat{g}_r(\psi) = 0 $ for all $\psi \in C_0^{\infty}(N_2 - N_1)$.  Thus $(\supp \hat{g}) \cap (N_1 - N_2) = \emptyset.$  We have thus proved that

\begin{equation}
f(k_1)\ne f(k_2) \ \text{implies} \ k_1 - k_2 \not\in \supp \hat{g} \nonumber
\end{equation}
or
\begin{equation}
k_1 -k_2 \in \supp \hat{g} \ \text{implies} \ f(k_1) = f(k_2)  \nonumber
\end{equation}
or
\begin{equation}
\supp \hat{g} \subset P_f: = \  \text{the set of periods of } f
\end{equation}
It follows that either $\supp \hat{g} = \{0\}$ in which case $g$ is constant, or $f$ is periodic.  We have already assumed that $f$ is not constant; now assume $g$ is not constant.  Then $f$ has a smallest positive period, $\tau_f$, and 
$$ \supp \hat {g} \subset \tau_f \mathbb{Z}.$$
If we choose $\phi \in C_0^{\infty}((-\tau_f/2,\tau_f/2))$ with $\phi = 1$ in a neighborhood of the origin, then there is a sequence of integers $\{n_j\}$ such that $\hat {g} = \sum_j \phi_j \hat{g}$ where $\phi_j(\xi)= \phi(\xi - n_j \tau_f)$.  Thus $\hat {g} = \sum_j c_j \delta_{n_j\tau_f}$ with $\delta_k$ the Dirac delta at $k$. It follows that $g= \sum_j c_j'e^{i n_j \tau_f x}$ which implies $g$ has period $\tau_g= 2\pi/\tau_f$.  This proves the result when $f$ and $g$ are smooth with bounded derivatives.  In the general case convolve $f$ and $g$ with the gaussian $\delta_0^{\sigma}(x) = \frac{1}{\sqrt{2\pi} \sigma}e^{-|x|^2/2\sigma^2}$.  The resulting functions of $P$ and $Q$ still commute.  Since 
\begin{align*}
\supp (\widehat{\delta_0^{\sigma}*f})&= \supp \hat {f} \\
\supp (\widehat{\delta_0^{\sigma}*g})&= \supp \hat {g},
\end{align*}
it follows that $\delta_0^{\sigma}*f$ and $f$ have the same periods as do  $\delta_0^{\sigma}*g$ and $g$.  This completes the proof.
\end{proof}

\begin{remark}
Actually the result is true without the requirement that the functions $f$ and $g$ be real.  The complex case follows from the result just proved and a result of Fuglede \cite {Fu1} (see also \cite {R} for a very simple proof)  which states that if a normal operator commutes with another operator then so does its adjoint.  Thus the real and imaginary parts of $f$ and $g$ commute with one another.  With this information the proof is straight-forward.  
\end{remark}

\section{Theorem  \ref{monotone} - monotonicity}

If we assume that $f$ and $g$ are absolutely continuous with derivatives in $L^1(\R)$, the proof of monotonicity is given by Kato in \cite {TK} and is straightforward. We give a sketch of the proof:  The integral kernel of the positive operator $ i[f(P), g(Q)]$ is 
 $$K(x,y) = (\sqrt{2\pi})^{-1}\frac{g(x) -g(y)}{x-y} \widehat {f'}(y-x).$$  With $[f] = f(\infty) - f(-\infty)$
 the condition of positivity implies  $K(x,x) = g'(x)[f]/2\pi \ge 0$ and $K(x,x)K(y,y) \ge |K(x,y)|^2$ or $ g'(x) g'(y)[f]^2 \ge (2\pi)^2 |K(x,y)|^2$.  Thus unless $K$ is identically zero, $[f]  \ne 0$ which implies $g$ is monotone.   Using complex conjugation, $C_0$, and the Fourier transform, $\mathcal{F}$, we note that 
  \begin{equation}\label{f&g}
 \mathcal{F}^{-1}C_0 i[f(P),g(Q)] C_0^{-1}\mathcal{F} = i[g(P),f(Q)]
 \end{equation}
 which is therefore positive. Thus the same argument gives that $f$ is also monotone and it follows they are both either increasing or both decreasing. 
  
\begin{proof} [Proof of Theorem \ref{monotone}]
We first assume that $f$ and $g$ are infinitely differentiable with bounded derivatives.  In the proof we use the distribution kernel of the operator $C = $ given by 

$$ K(x,y) = (\sqrt{2\pi})^{-1}\frac{g(x) -g(y)}{x-y} \widehat{f'}(y-x)$$\\
We do not know that $f'$ is integrable which accounts for the distribution nature of the kernel.  

Since the continuity of the kernel is not yet known, we use a smeared out version of the inequality $|K(x_1,x_2) |^2 \le  K(x_1,x_1) K(x_2,x_2) $.  Put 

$$ \delta_a^{\sigma}(x) = \frac{1}{\sqrt{2\pi} \sigma}e^{-|x-a|^2/2\sigma^2}$$\\
and note that 
$$\delta_a^{\sigma}(x)\delta_a^{\sigma}(y)= \delta_0^{\sqrt 2\sigma}(x-y)\delta_a^{\sigma/\sqrt 2}((x+ y)/2).$$\\
For $\phi \in C_0^{\infty}(\mathbb{R})$, let 
\begin{align}
I_{\sigma} &= (\delta_0^{\sigma}*\phi,C\delta_0^{\sigma}*\phi) = \int (\delta_{x_1}^{\sigma},C\delta_{x_2}^{\sigma})\overline{\phi(x_1)}\phi(x_2)dx_1 dx_2 
\end{align}
and note $$ |(\delta_{x_1}^{\sigma},C\delta_{x_2}^{\sigma})| \le (\delta_{x_1}^{\sigma},C\delta_{x_1}^{\sigma})^{1/2}(\delta_{x_2}^{\sigma},C\delta_{x_2}^{\sigma})^{1/2}$$  so that
\begin{align}
I_{\sigma} &\le (\int  (\delta_{x}^{\sigma},C\delta_{x}^{\sigma})^{1/2}|\phi(x)| dx)^2 \le \int_{\supp \phi} 1dx \int  (\delta_{x}^{\sigma},C\delta_{x}^{\sigma})|\phi(x)|^2 dx \nonumber \\
&= c\int \delta_0^{\sqrt 2\sigma}(x-y)K(x,y)\delta_0^{\sigma/\sqrt 2}((x+ y)/2 - x_1)|\phi(x_1)|^2 dx_1dxdy \nonumber\\
&=c\int \delta_0^{\sqrt 2\sigma}(x-y)K(x,y)\delta_0^{\sigma/\sqrt 2}*|\phi|^2((x+y)/2) dxdy \nonumber\\
&=:cJ_{\sigma}/\sqrt{2\pi}
\end{align}
where $c=  \int_{\supp \phi} 1dx$.  We abbreviate $\psi_{\sigma} = \delta_0^{\sigma/\sqrt 2}*|\phi|^2$ and use Taylor's theorem to write 

\begin{align}
\frac{g(y+t) - g(y)}{t} &= g'(y) + t G_1(y,t)\\
 G_1(y,t) &= \int_0^1 g^{(2)}(y+\theta t) (1-\theta) d\theta.
\end{align}
Thus we obtain
\begin{align}
J_{\sigma}&= \int \delta_0^{\sqrt 2\sigma}(t) g'(y) \widehat {f'}(-t)\psi_{\sigma}(y+ t/2) dt dy \nonumber \\
&+ \int t \delta_0^{\sqrt 2\sigma}(t) G_1(y,t)\widehat {f'}(-t)\psi_{\sigma}(y+ t/2) dt dy
\end{align}
and replacing $\psi_{\sigma}(y+ t/2)$ by $\psi_{\sigma}(y)$ we pick up another error term so that 
\begin{align}
J_{\sigma}&= \int \delta_0^{\sqrt 2\sigma}(t) g'(y) \widehat {f'}(-t)\psi_{\sigma}(y) dt dy \nonumber \\
& + \int t \delta_0^{\sqrt 2\sigma}(t)G_2^{\sigma}(y,t) \widehat {f'}(-t) dt dy \\
&G_2^{\sigma}(y,t) = g'(y)(\psi_{\sigma}(y+ t/2) - \psi_{\sigma}(y))/t + G_1(y,t) \psi_{\sigma}(y+t/2).
\end{align}
We now have 
\begin{align}
J_{\sigma}& = J_{\sigma}^1 + J_{\sigma}^2 \nonumber \\
& J_{\sigma}^1 = \left(\int \delta_0^{\sqrt 2\sigma}(t) \widehat {f'}(-t)dt \right)\left(\int g'(y)\psi_{\sigma}(y)dy\right) \nonumber \\
&  J_{\sigma}^2 = \int I_2^{\sigma}(y) dy \nonumber \\
& I_2^{\sigma}(y) = \int G_2^{\sigma}(y,t) t  \delta_0^{\sqrt 2\sigma}(t)\widehat {f'}(-t)dt = \int F^{\sigma}_y(\xi)f(\xi) d\xi \nonumber \\
&F_y^{\sigma}(\xi) = (-i/\sqrt{2\pi})\int G_2^{\sigma}(y,t)t^2 \delta_0^{\sqrt 2\sigma}(t)e^{i\xi t}dt.
\end{align}
We estimate $D^m\psi_{\sigma} = \delta_0^{\sigma/\sqrt 2}*D^m|\phi|^2$ to find for all $m$ and $n$
\begin{align*}
|D^m\psi_{\sigma} (x)| \le c_{n,m} (1+|x|)^{-n}
\end{align*}
uniformly for  $\sigma \in (0,1)$ .
We use $G_2^{\sigma}(y,t)= \int_0^{1/2} \psi'_\sigma(y+ \theta t) d\theta g'(y) + G_1^{\sigma}(y,t) \psi_{\sigma}(y+ t/2) $ and obtain
\begin{align*}
|G_2^{\sigma}(y,t)| \le c_n((1+|y+t/2|)^{-n} + (1+|y|)^{-n}).
\end{align*}
Thus for $\sigma \in (0,1)$
\begin{align}\label{sigmasquaredbound}
|F_y^{\sigma}(\xi)|&\le c_n'\int ((1+|y+t/2|)^{-n} + (1+|y|)^{-n})t^2 \delta_0^{\sqrt 2\sigma}(t)dt \nonumber \\
&\le d_n \sigma^2(1+|y|)^{-n}.
\end{align}
We also need some decay of $F_y^{\sigma}(\xi)$ in $\xi$.  Thus

\begin{align} \label{integral}
|(1+ \xi^2)F_y^{\sigma}(\xi)| \le c\int |(1-d^2/dt^2)G_2^{\sigma}(y,t)t^2 \delta_0^{\sqrt 2\sigma}(t)|dt.
\end{align}
We easily see that $|d^m/dt^m G_2^{\sigma}(y,t)| \le c_m'$ uniformly for $\sigma \in(0,1)$. Thus differentiating $t^2 \delta_0^{\sqrt 2\sigma}(t)$ and integrating in (\ref{integral}) we find
\begin{align}\label{derivbound}
(1+\xi^2) |F_y^{\sigma}(\xi)| \le c'.
\end{align}
Interpolating between (\ref{sigmasquaredbound}) and (\ref{derivbound})  we obtain for $\theta \in (0,1)$
\begin{align*}
 |F_y^{\sigma}(\xi)|\le c''_n(1+\xi^2)^{-(1-\theta)}\sigma^{2\theta}(1+|y|)^{-n\theta}.
\end{align*}
It thus follows that for $\theta \in (0,1/2)$ and any $n$
\begin{align*}
|I_2^{\sigma}| \le c_{n,\theta} \sigma^{2\theta} (1+|y|)^{-n\theta}
\end{align*}
which gives 
\begin{align*}
|J_2^{\sigma}| \le c_{\theta} \sigma^{2\theta}.
\end{align*}
We have thus shown
\begin{align} \label{Jsigma}
J_{\sigma}/\sqrt{2\pi} = \int (\delta_x^{\sigma}, C\delta_x^{\sigma})|\phi(x)|^2 dx= \left(\int  \delta_0^{\sqrt 2\sigma}(t) \widehat {f'}(-t)dt \right) \left(\int g'(y) \psi_{\sigma}(y) dy \right) + O( \sigma^{2\theta})
\end{align}
for $\theta \in (0,1/2)$.
Notice that 
\begin{align}\label{beta sigma}
\beta(\sigma): = \sqrt{2\pi}\int  \delta_0^{\sqrt 2\sigma}(t) \widehat {f'}(-t)dt = \int e^{-(\sigma \xi)^2}f'(\xi) d\xi
\end{align}
is independent of $\phi$.  Let us assume that $g'$ is not of constant sign.  Choose $\phi_1 \in C_0^{\infty}(\mathbb R)$ so that $\int g'(y) |\phi_1(y)|^2 dy >0$.  Then with $\psi^1_{\sigma} =  \delta_0^{\sigma/\sqrt 2}* |\phi_1|^2$ we have \\
$\lim _{\sigma \to 0} \int g'(y)\psi^1_{\sigma}(y) dy  = \int g'(y) |\phi_1(y)|^2 dy >0$ and thus 

\begin{align}
\beta(\sigma) = J_{\sigma}\big(\int g'(y)\psi^1_{\sigma}(y) dy\big)^{-1} + O(\sigma^{2\theta})
\end{align}
which implies  that for small $\sigma$, $\beta(\sigma) \ge - c_1 \sigma^{2\theta}$ for some $c_1 >0$ since $J_{\sigma} \ge 0$.  Similarly choosing $\phi_2 \in C_0^{\infty}(\mathbb R)$ so that $\int g'(y) |\phi_2(y)|^2 dy <0$ we find for small $\sigma$
 $\beta(\sigma) \le c_2 \sigma^{2\theta}$ for some $c_2 >0$. It follows that for small $\sigma$, for any $\phi\in C_0^{\infty}(\mathbb R)$, $J_{\sigma} = O(\sigma^{2\theta})$.  But since $I_{\sigma} = (\delta_0^{\sigma}*\phi,C\delta_0^{\sigma}*\phi) \le cJ_{\sigma}/\sqrt{2\pi}$ we can take the limit $\sigma \to 0$ and conclude $(\phi,C\phi) = 0$ for all $\phi\in C_0^{\infty}(\mathbb R)$ and thus $C=0$.  Hence if $C\ge 0$ and $C \ne 0$, $g$ is a monotone function.  By (\ref{f&g}) the same holds for $f$.
Combining (\ref{Jsigma}) and (\ref{beta sigma}) we see that 

\begin{align*}
I_{\sigma}&= (\delta_0^{\sigma}*\phi,C\delta_0^{\sigma}*\phi) \le cJ_{\sigma} /\sqrt{2\pi}\\
&\le c(1/\sqrt{2\pi})\left(\int e^{-(\sigma \xi)^2}f'(\xi) d\xi \right) \left(\int g'(y) \psi_{\sigma}(y) dy\right) + O( \sigma^{2\theta}).
\end{align*}
Taking $\sigma \to 0$ we obtain
\begin{align*}
(\phi,C \phi)\le c(1/\sqrt{2\pi})\int f'(\xi) d\xi \int g'(y) |\phi(y)|^2dy
\end{align*}
(where again $c = \int_{\supp \phi} 1dx$) so that $g'(x)f'(\xi) \ge 0$. 

To deal with the general case where neither $f$ nor $g$ is known to be smooth note that if $i[f(P),g(Q)] \ge 0$ and non-zero the same is 
true if we replace $g$ with $g*\rho$ where $\rho$ is a smooth non-negative approximation to the identity.  Similarly for $f$ replaced with $f*\zeta$ .  It then follows that both $g*\rho$ and $f*\zeta$ are monotone.  By taking $\rho_n$ a suitable sequence, $g*\rho_n \to g$ on a set $E$ of full Lebesgue measure and thus $g$ is monotone on $E$.  If we take $G(x) = \lim _{u\in E, u\uparrow x} g(u)$, then $G$ is monotone and equals $g$ a.e.  Clearly the same idea works for $f$.  It is easy to see that $G$ and $F$, the corresponding version of $f$, are either both increasing or both decreasing.
\end{proof}

 \section{Theorem 1.5 - Continuity of $\mathrm {g}$, absolute continuity of the inverse of $\mathrm{g}$ }
 
In this section we assume that the commutator $ i[f(P),g(Q)]$ is non-zero and positive.

\begin{lemma} \label{kernel2}
Suppose  $ i[f(P),g(Q)] = C$ where $C$ is positive and $f$ and $g$ are monotone non-decreasing and bounded.  Then for all $\psi \in L^2(\R)$, $( i[f(P), g(Q)]\psi)(x) =  \int K(x,y) \psi(y) dy$  where 
\begin{equation} \label{kernel1} 
K(x,y) = \frac{1}{\sqrt{2\pi}} \frac{g(x) - g(y)}{x-y} \widehat {df}(y-x).
\end{equation}   The operator $C$ is trace class with $\tr C = [f][g]/2\pi$, and the kernel $K$ 
is square integrable.
\end{lemma}

\begin{proof}
Let $g_t = \phi_t * g$ where $\phi$ is a non-negative smooth function of compact support whose integral is $1$ and $\phi_t(x) = t^{-1} \phi(t^{-1}x)$.  Then
$$ i[f(P),g_t(Q)] = \int_{\mathbb{R}} C(a) \phi_t(a) da$$ 
where $C(a) = e^{-iPa}Ce^{iPa}$.  $ i[f(P),g_t(Q)] $ is a positive operator with continuous kernel  so we can calculate its trace which is easily seen to be $[f][g]/2\pi $.  We have
$$\tr \int_{\mathbb{R}} C(a) \phi_t(a) da  =  \int (\tr e^{-iPa}Ce^{iPa})\phi_t(a)da = \int \tr (C) \phi_t(a) da = \tr C$$ 
which shows $\tr C = [f][g]/2\pi $. Since $ i[f(P),g_t(Q)]$ has constant trace its Hilbert-Schmidt norm is uniformly bounded.  Thus if we abbreviate \\ $K_t (x,y) = \int K(x-a,y-a) \phi_t(a) da$, we have 
$$\infty > \liminf_{t \downarrow 0} \int |K_t(x,y)|^2 dx dy \ge \int |K(x,y)|^2 dx dy$$ 
by Fatou's lemma.  Here we have used that $K_t \rightarrow K$ a.e.  It follows easily that  $K_t \rightarrow K$ in $L^2(\R^2)$.  For an $L^2(\R)$ function $\psi$, let $F_t = \int K_t(x,y)\psi(y) dy $ and $F = \int K(x,y)\psi(y) dy$.  We have $||F_t - F|| \le ||K_t - K||_{L^2(\R^2)} ||\psi ||$.  But $F_t \rightarrow C \psi$ strongly which gives (\ref{kernel1}).
This finishes the proof.

\end{proof}

\begin{lemma}
Suppose $i[f(P),g(Q)]=C$ where $C$ is non-zero and positive while $f$ and $g$ are bounded and monotone increasing.  Then $g$ is continuous.
\end{lemma}

\begin{proof}

We use the square integrability of the kernel $K$.  Suppose $g$ has a jump of $k'$ at $0$.  Then for all $\epsilon >0$ 
$$\int_ {x<-\epsilon, y>\epsilon} |(g(y) - g(x))/(y-x)|^2 |\widehat{df}(y-x)|^2 dxdy < k$$ for some $k$.  Now 
$|\widehat{df}(y-x)| \ge \delta > 0$ for $y-x$ small while $g(y) - g(x) \ge  k' >0$ in the range of integration. This contradicts the above inequality.

\end{proof}

\begin{lemma}
Suppose  $ i[f(P),g(Q)] = C$ where $C$ is positive and $f$ and $g$ are monotone non-decreasing and bounded.  Then $g$  is strictly increasing.  

\end{lemma}

\begin{proof}
Suppose $g$ is constant on $I = (a,b)$.  Then by the explicit form of the kernel $K$ of the commutator $C =  i[f(P),g(Q)]$, we see that for $\psi \in C_0^{\infty}(I), (\psi,C\psi) =0$.  Since $C\ge 0$, this means $C\psi = 0$ which in turn implies $(\frac{g(x) - g(y)}{x-y}) \widehat{df}(y-x) = 0 $ for $y \in I$ and $x \in \mathbb R$.  But  $\widehat{df}$ is continuous and non-zero at $0$, so non-zero in an interval $J = (-c,c), c>0$.  Thus $g$ is constant on the interval $I + J$.  The result follows by induction.

\end{proof}

We give a quick proof of Putnam's theorem \cite{CFWK, P} on positive commutators:

\begin{thm} (Putnam)
Suppose $A$ and $H$ are self-adjoint operators with $A$ bounded.  Suppose $i[H, A] = C$ with $C$ self-adjoint and positive. We assume this equality is true in the sense that for each $\psi \in D(H)$ 
$$ i(H\psi, A \psi) - i(A\psi, H\psi) =  ||C^{1/2}\psi ||^2.$$
Then the absolutely continuous subspace of $H$, $\mathcal{H}_{ac}(H) \supset \overline{ \Ran{C}}$.
\end{thm}

\begin{remark}
It will be clear in the proof that $A$ need not be bounded but it rather  suffices that $D(A)$ contain all vectors in the range of $E_{H}(J)$ for all finite intervals $J$.  (Here $E_H(\cdot)$ is the projection valued spectral measure associated with the operator $H$.)  This generalization may have limited applicability so we do not include a proof (although it should be obvious).
\end{remark}

\begin{proof}
For a finite interval $I = (a, b)$, let $\lambda = (a + b)/2$, and $E = E_H(I)$.  We use the notation $|J|$ for the Lebesgue measure of the Borel set $J$. Then note 
\begin{align}
|(E\psi, i[H,A]E\psi)|& = |(E\psi, i[H - \lambda I,A]E\psi)|\\
&\le  |((H-\lambda)E\psi, AE \psi) - (AE\psi, (H-\lambda)E\psi)| \\
&\le 2 |I|/2 ||A|| ||\psi||^2.
\end{align}
Thus $$||ECE|| \le |I| ||A||.$$
We also have 
$$||ECE|| = ||(EC^{1/2})(EC^{1/2})^*|| = ||(EC^{1/2})^*(EC^{1/2})|| = ||C^{1/2}EC^{1/2}||$$
so that if $||\psi|| = 1$,
$$(C^{1/2}\psi,E_H(I)C^{1/2}\psi) \le |I| ||A|| .$$
This inequality extends from finite open intervals, $I$, to arbitrary Borel sets and shows that  $\mathcal{H}_{ac}(H) \supset \Ran{C^{1/2}}$.  But it is easily seen that $\overline{  \Ran{C^{1/2}}} = \overline{\Ran{C}}$ which proves the theorem. 

\end{proof} 

We use Putnam's theorem to show absolute continuity of the inverse of the function $g$.
\begin{proposition}
Suppose $i[f(P),g(Q)] = C$, where $f$ and $g$ are real increasing bounded functions and $C$ is positive and non-zero.  Then the inverse function $g^{-1}$ is absolutely continuous.  
\end{proposition}

\begin{proof}
Let $M_g$ be multiplication by $g$, that is $M_g = g(Q)$.  From Putnam's theorem we have that $\mathcal{H}_{ac}(M_g) \supset{\Ran C}$ and since for $a \in \mathbb R$ $$i[f(P-aI),g(Q)] = e^{iaQ}Ce^{-iaQ}$$ it follows that  $\mathcal{H}_{ac}(M_g) \supset{e^{iaQ}\Ran C}$.  Suppose $(\phi, e^{iaQ}C\psi) = 0$ for all $a \in \mathbb R$ and all $\psi \in L^2(\mathbb R)$.  It follows that $\phi(x) (C\psi)(x) = 0$ a.e.  Choose $a > 0$ so that $\text{Re} \ \widehat {df}(x) > 0$ for $x \in I = (-a,a)$.   Choose a non-negative  $\psi \in C_0^{\infty}(\frac{1}{2}I)$ with $\psi(0) >0$. Then for any $x_0 \in \mathbb R$, $ \text {Re}(C\psi (x)) = (2\pi)^{-1/2} \int \frac{g(x) - g(y)}{x-y}\text{Re}(\widehat{df}(y-x))\psi(y -x_0) dy > 0$ for $x \in \frac{1}{2}I + x_0$.  Thus for all $x_0$, $\phi (x) = 0 $ for  a.e. $x \in \frac{1}{2}I + x_0$, or in other words $\phi$ is the zero vector in $L^2(\mathbb R)$.  It follows that $\mathcal{H}_{ac}(M_g) = L^2(\mathbb R)$.
 
With $H = M_g$ we have that $E_H(J)$ is multiplication by $1_{\{x:g(x) \in J\}} = 1_{g^{-1}(J)}$.  So $M_g$ being absolutely continuous as an operator means that if the Borel set $J$ has  Lebesgue measure  $0$ then $E_{M_g}(J) =0$ and thus $g^{-1}(J)$ has Lebesgue measure $0$.  Thus the strictly increasing function $g^{-1}$ is absolutely continuous.  (Note that even with the fact that $g$ is strictly increasing, the absolute continuity of $M_g$ as an operator alone does not imply that $g$ is continuous.)
\end{proof}
We remark that because of (\ref{f&g}) all the results of this section apply to $f$ as well as $g$.
\section{Theorem 1.6 - $(K_b, L^{\infty}) \xrightarrow {?}(K_b, K_{\hat b})$}

\begin{proof}[Proof of Theorem 1.6] 

We assume $g$ has the integral representation (\ref{tom})

where $\mu$ is a non-zero finite positive measure with $\int |t|e^{2\hat b t}d\mu(t) < \infty$.  The proof where  $\int |t|e^{-2\hat b t}d\mu(t) < \infty$ is similar.  We are aiming to prove that $f$ has the integral representation $$f(x) = \int \tanh b(x-t) d\nu(t) + d$$ for some finite positive measure $\nu$ and with $b\hat b =\pi/2$.  We first assume that $f$ is as above but with $d\nu(t)$ replaced with $w(t) dt $ where $w\in L^1(\R) \cap L^{\infty}(\R)$ and where $w$ is real and continuous but not necessarily non-negative.  This will be justified at the end of the proof.  Our first task is to prove that $w\ge 0$.  For
the purpose of obtaining a contradiction suppose that is not the case.  Then
the sets $E_1 = \{s: \, w(s) > 0\}$ and $E_2 = \{s: \, w(s) < 0\}$ are both 
non-empty.  For $j = 1,2,$ put $w_j = 1_{E_{j}}\cdot |w|$, where $1_F$ 
denotes the characteristic function of the set $F$.  Clearly, $w = w_1-w_2$.  A computation based on (\ref{kernel}) 
shows
that if $\psi$ is in $L^2(\R)$, then
\begin{equation} \label{windowedFT}
(\psi, i[f(P),g(Q)]\psi)
= {\hat b \over \pi} \int\int \left|\int {e^{-isx} \over \cosh \hat b(x-t)} \ \psi
(x) dx \right|^2 d\mu (t) \, w(s) ds \ge 0.
\end{equation}
or    equivalently
\begin{eqnarray*}
\lefteqn{\hspace*{-30pt}
\int\int\left|\int \frac{e^{isx}}{\cosh \hat b(x-t)} \ \psi(x)dx\right|^2
w_2(s)ds \, d\mu(t)} \\[5mm]
& \leq & \int\int\left|\int \frac{e^{isx}}{\cosh \hat b(x-t)} \  \psi(x)dx\right|^2
w_1(s)ds \, d\mu(t),
\end{eqnarray*}
for all $\psi \in L^2(\R)$ (we find it convenient to change signs in the
complex exponential in (\ref{windowedFT}), and so have replaced $\psi$ by $\bar{\psi}$).  The
last inequality can be written as
\begin{equation}\label{Ajineq}
\|A_2\psi\| \leq \|A_1\psi\|,
\end{equation}
with $A_j:  \ L^2(\R) \rightarrow L^2(\nu_j \times \mu)$  defined by
$$
(A_j \psi)(x,t) = \int^\infty_{-\infty} \frac{e^{isx}}{\cosh\hat b(x-t)}
\psi(x)dx
$$
for $j = 1$ or 2, where $d\nu_j(s) = w_j(s)ds$.  These operators are in fact bounded, as one can see
by applying (\ref{windowedFT}) with $\nu = \nu_1$ or $\nu = \nu_2$ and with $f$
correspondingly modified. We see (writing $A$ for $A_1$ or $A_2$ as the
case may be)
\begin{eqnarray*}
\frac{\hat b}{\pi} \|A\psi\|^2
& = & (\bar{\psi},i[f(P),g(Q)]\bar{\psi}) \\[3mm]
& \leq & 2 \|f\|_\infty \|g\|_\infty \|\psi\|^2.
\end{eqnarray*}
Let us turn to the inequality (\ref{Ajineq}), which can be restated as $A^*_2 A_2
\leq A^*_1 A_1$.  According to a lemma of Douglas [4], $A_2 = CA_1$ for
some contraction operator $C: \, L^2(\nu_1 \times \mu) \rightarrow L^2(\nu_2 \times
\mu)$, so $A^*_2 = A^*_1 C^*$ and Ran $A^*_2 \subset$ Ran $A^*_1$. We 
will show that this range inclusion cannot occur, and to this end look closely at the 
operators $A^*_j$.  Note
that if $h \in L^2(\nu_j \times \mu)$,
$$
(A^*_j h)(x) = \int\int \frac{e^{-isx}}{\cosh \hat b(x-t)} h(s,t) w_j(s)ds \,
d\mu(t).
$$
It follows that 
$$2^{-1}e^{\hat b x}(A_j^*h)(x) = \int \int e^{-isx}(1- e^{2\hat b(t-x)}(1 + e^{2\hat b(t-x)})^{-1})h(s,t) w_j(s) ds e^{\hat b t}d\mu(t).$$ 
Further 
$$e^{2\hat b(t-x)}(1 + e^{2\hat b(t-x)})^{-1} = 1_{(-\infty,t)}(x) - U(x-t), \  \text{with}  \\\ |U(x)|\le e^{-2\hat b |x|}.$$
We see therefore that 
$$2^{-1}e^{\hat b x}(A_j^*h)(x) = I_j(x) + II_j(x)+ III_j(x)$$
where 
\begin{align}
I_j(x)& = \int\int e^{-isx} h(s,t)w_j(s)ds e^{\hat b t}d\mu(t) \\
II_j(x) & = - \int\int e^{-isx} 1_{(-\infty,t)}(x) h(s,t) w_j(s) ds e^{\hat b t}d\mu(t)\\
III_j(x)& = \int\int e^{-isx}U(x-t) h(s,t) w_j(s)ds e^{\hat b t} d\mu(t).
\end{align}
First consider $I_j$.  We define 

$$q_j(s) = \big(\int h(s,t)e^{\hat b t}d\mu(t)\big) w_j(s).$$  Then $q_j \in L^2(\R)\cap L^1(\R)$.  Indeed 
\begin{align*}
\int |q_j(s)|^2ds &= \int \Big |\int h(s,t)e^{\hat b t }d\mu(t) \Big |^2 w_j(s)^2 ds \\
 &\le ||w||_{\infty} \int e^{2\hat b t} d\mu(t) \int\int |h(s,t)|^2w_j(s)ds d\mu(t) < \infty 
\end{align*}
and 
\begin{align*}
&\int |q_j(s)|ds \le  \\
 &\int \int |h(s,t)|w_j(s)^{1/2} e^{\hat b t}w_j(s)^{1/2} d\mu(t) ds\\
 & \le \left(\int \int |h(s,t)|^2 w_j(s)d\mu(t)ds\right)^{1/2} \left(\int \int e^{2\hat b t } w_j(s)d\mu(t) ds \right)^{1/2}\\
& = ||w_j||_1^{1/2} \left (\int e^{2\hat b t} d\mu(t) \right )^{1/2} \left(\int \int |h(s,t)|^2 w_j(s) ds d\mu(t)\right)^{1/2} < \infty.
\end{align*}
Here we have used the Schwarz inequality and the facts that $w_j \in L^1(\R) \cap L^{\infty}(\R), \\ \int e^{2\hat b t} d\mu(t) < \infty,$ and $ h \in L^2 (\nu_j \times \mu)$.  We see that $I_j = \sqrt{2\pi} \hat {q_j}$ which implies $I_j \in L^2(\R)$ as well as in $C_0(\R)$, the space of continuous functions on $\R$ vanishing at infinity.

We turn to $II_j$.  We first write  $g(t,x) = \int e^{-isx}h(s,t)w_j(s) ds$ and note that by the Plancherel theorem 
\begin{equation}\label{plancherel}
\int |g(t,x)|^2 dx d\mu(t) \le 2\pi||w_j||_{\infty} ||h||_{L^2(\nu_j \times \mu)}^2
\end{equation}
We have 
\begin{align*}
II_j(x)& = - \int_{(-\infty, 0)}1_{(-\infty,t)}(x) g(t,x)e^{\hat b t}d\mu(t)\\
& - \int_{[0,\infty )}1_{(-\infty,0)}(x)g(t,x)e^{\hat b t}d\mu(t)\\
& - \int_{[0,\infty)}1_{[0,t)}(x)g(t,x)e^{\hat b t}d\mu(t)
\end{align*}
Clearly from (\ref{plancherel}) and the Schwarz inequality each of these three terms is in $L^2(\R)$.  Evidently the first two terms vanish for $x>0$ so that their sum can be written as a Fourier transform, $\hat{v}_j$, with $v_j \in H^2_-$, the Hardy space of the lower half plane \cite{Duren}.  (Here we follow the convention of identifying elements of $H^2_-$, which are analytic in the open lower half plane, with their boundary value functions which lie in $L^2(\R)$.)  Using the condition $\int_{[0,\infty)} te^{2\hat b t}d\mu(t) < \infty$ and the Schwarz inequality we see that the third piece of $II_j$ is also in $L^1(\R)$, thus of the form $\hat p_j$ with $p_j \in C_0(\R) \cap L^2(\R)$.  In summary, $II_j = \hat v_j + \hat p_j$ with $v_j \in H^2_-$ and $p_j \in C_0(\R) \cap L^2(\R)$.  Finally $III_j \in L^2(\R)$ by essentially the same argument that shows each term of $II_j \in L^2(\R)$.  Moreover $III_j\in L^1(\R)$ which can be seen as follows.  On rewriting $III_j(x)$ in terms of $g(t,x)$ defined above we have

\begin{align*}
|III_j(x)|^2& \le \int |U(x-t)|^2 e^{2\hat bt} d\mu(t)\int |g(t,x)|^2 d\mu(t)\\
\text{and thus}&\\
\int |III_j(x)| dx & \le ||U||_2\left (2\pi ||w_j||_{\infty}||h||_{L^2(\nu_j \times \mu)}^2 \int e^{2\hat b t}d\mu(t)\right)^{1/2}
\end{align*}
where we have used (\ref{plancherel}). Thus $III_j =\hat u_j$ where $u_j \in C_0(\R) \cap L^2(\R)$.

On putting the above together we have 
\begin{equation}\label{ebxeqn}
2^{-1} e^{\hat b x} (A_j^*h)(x) = (\hat{q}_j+ \hat{v}_j + \hat {p}_j+ \hat{u}_j)(x).
\end{equation}
To obtain a contradiction we now use the fact that 
\begin{equation} \label{rangeinclusion}
\text{Ran} A_2^* \subset \text{Ran} A_1^*.  
\end{equation}
Choose $h_2 \in L^2(\nu_2\times \mu)$ as $h_2(s,t) = 1_F(s)$, where $F$ is a non-degenerate compact interval contained in the (assumed) non-empty open subset $E_2$. By (\ref{rangeinclusion}), there exists $h_1 \in L^2(\nu_1\times\mu)$ with 

\begin{equation} \label{equality}
A_2^*h_2 =A_1^* h_1.
\end{equation}
For $j=1$ or $2$, associate $q_j,v_j,p_j,$ and $u_j$ to $h_j$ as they are associated to $h$ in (\ref{ebxeqn}).  By (\ref{equality}) we have 
$$\hat{q}_2+ \hat{v}_2 + \hat{p}_2 + \hat{u}_2 = \hat{q}_1+ \hat{v}_1 + \hat{p}_1 + \hat{u}_1,$$
or applying the inverse Fourier transform and rearranging
\begin{equation}\label{v1-v2}
q_2 -q_1 + p_2 -p_1 + u_2-u_1 = v_1 -v_2.
\end{equation}
Clearly $$q_2(s)  = 1_F(s) w_2(s)\int e^{\hat b t} d\mu(t) $$
has non-trivial jump discontinuities at the end points of $F$ since $w_2$ is positive and continuous on $E_2$.  Let $J \subset E_2$ be a compact interval containing one  of the endpoints of $F, x_0$, in its interior.  Since $q_1 \equiv 0$ on $E_2$ and $p_1,p_2,u_1,u_2$ are continuous, $v_1 - v_2 $ has a jump discontinuity at $x_0$ .  However this will contradict a theorem of Lindel\"of, since $v_1 - v_2 \in H_-^2$.  Indeed the Poisson integral provides the analytic extension of $v_1-v_2$ into the lower half plane.  The left side of (\ref{v1-v2}), and thus  $v_1-v_2$ are bounded in a neighborhood of $J$ in $\R$, so this analytic extension is bounded in the open half disk $B$ in the lower half plane with the interior of $J$ as its diameter.  It follows that if $\phi$ is a conformal map of the open unit disk $\mathbb{D}$ onto $B$ (such a $\phi$ must extend to a homeomorphism of $\bar{\mathbb{D}}$ onto $\bar{B}$) then $(v_1-v_2)\circ \phi$ is a bounded analytic function in $\mathbb{D}$ whose boundary value function has a jump discontinuity at $\phi^{-1}(x_0)$.  This contradicts Lindel\"of's theorem \cite {L} (see also \cite {G}) and shows that indeed $w\ge 0$.

Finally, we consider general $f$ in $L^\infty(\R)$.  Suppose our hypothesis
$i[f(P),g(Q)] \geq 0$ is in effect.  Let
$$
u_r(x) = {1 \over r\pi} \left({\sin rx \over x}\right)^2,
$$
an approximate identity on $\R$ as $r \rightarrow \infty$.  Since $u_r\ge 0$ it follows that \\ $i[f*u_r(P),g(Q)] \geq 0$.  Let $f_r = f*u_r$.  Since $f_r$ is smooth, it also follows from Theorem \ref{monotone} that
$(f_r)' \geq 0$.  Since $f_r$ is bounded, we then see $(f_r)'$ is
in $L^1(\R)$.  By direct computation $\widehat{u_r}$ has compact support, and so
$\widehat{f_r}$ (taken in the sense of tempered distributions) does as well.
The same must be true for the Fourier transform of $(f_r)'$ which is a continuous function.

Now we let $\phi(x) = \tanh bx$ and note

\begin{equation} \label{FT}
\widehat{\phi'}(y) = \hat b \sqrt{{2 \over \pi}} {y \over \sinh (\hat b y)}.
\end{equation}
Let $\sigma$ be a real even function in the Schwartz space agreeing with
$1/\widehat{\phi'}$ on the support of $\widehat{f_r'}$.  Then
$\widehat{f_r'} = \widehat{\phi'} \cdot \widehat{f_r'}\cdot \sigma,$
whence
\begin{equation}\label{2convolutions}
f_r' = {1 \over 2\pi} \phi' *f_r'*\widecheck{\sigma}.
\end{equation}
We put $w_r = {1 \over 2\pi} f_r'*\widecheck{\sigma}$,
a convolution of $L^1(\R)$-functions and so in $L^1(\R)$.
In particular,
\begin{eqnarray*}
\|w_r\|_1 & \leq & {1 \over 2\pi} \left\{\int^\infty_{-\infty}
(f_r)'(x)dx\right\}\|\ \widecheck{\sigma}\|_1 \\[4mm]
& = & {1 \over 2\pi}
\left\{(f*u_r)(\infty)-(f*u_r)(-\infty)\right\}\|\widecheck{\sigma}\|_1
\\[4mm]
& \leq & {1 \over \pi} \|f\|_\infty \|\widecheck{\sigma}\|_1.
\end{eqnarray*}
Further, $w_r$ is continuous and tends to zero at $\pm \infty$, hence
$w_r \in L^\infty(\R)$. On integrating equation (\ref{2convolutions}) we have
\begin{equation}\label{tanhrepforfr}
f_r(x) = \int^\infty_{-\infty} \tanh b(x-s) w_r(s)ds + d_r,
\end{equation}
where $d_r$ is a real constant.  Since $i[f*u_r(P),g(Q)] \geq 0$,
our previous results show that $w_r \geq 0$ a.e.
On letting $x \rightarrow \pm \infty$ in (\ref{tanhrepforfr}) we find by the dominated
convergence theorem that
$$
f_r(\infty) +f_r(-\infty) = 2d_r,
$$
and 
$$f_r(\infty) - f_r(-\infty) = 2||w_r||_1$$ \\
so that $|d_r| \leq \|f\|_\infty$ for $r > 0$. Since $f_r(\pm \infty) = f(\pm \infty)$ we have $\int w_r(s) ds = [f]/2$.   Now consider $w_r(s)ds$
as a measure on the two-point compactification $[-\infty,\infty]$ of
$\R$. Then there exists a finite positive measure $\nu$ on $\R$
and $\epsilon_1, \epsilon_2 \geq 0$, a real number $d$ and a sequence
$r_n \rightarrow \infty$ such that
$$
w_{r_{n}}(s)ds \rightarrow \nu + \epsilon_1 \delta_{-\infty} + \epsilon_2
\delta_{+\infty},
$$
in the weak-$*$ topology and $d_{r_{n}}\rightarrow d$ as $n \rightarrow
\infty$; here $\delta_{\pm \infty}$ denotes the unit point mass at $\pm
\infty$.  Then taking a limit along $r_n$ in (\ref{tanhrepforfr}) yields
\begin{eqnarray*}
f(x) & = & \lim_{n \rightarrow \infty} f_{r_{n}}(x) \\[4mm]
& = & \int^\infty_{-\infty} \tanh b(x-s)d\nu(s) +\epsilon_1 - \epsilon_2
+ d
\end{eqnarray*}
almost everywhere on $\R$.  It follows that $f \in K_{\hat b}$, and
the proof is complete.
\bigskip

\end{proof}

\section{Finite rank positive commutators}

Considering that the positive commutator is trace class it is interesting to consider the case where the positive commutator has finite rank.  The next lemma shows that, in particular, in considering the rank one case, Kato's assumption that $f$ and $g$ are absolutely continuous with $L^1(\R)$ derivatives can be proved. 
\begin{lemma}
Suppose the commutator  $C =i[f(P),g(Q)]$ is positive, non-zero, and has finite rank.  Then $f$ and $g$ can be taken differentiable with continuous derivatives.  These derivatives are everywhere non-zero.
\end{lemma}

\begin{proof}

We can assume that $f$ and $g$ are increasing.
The kernel of the commutator  $C =i[f(P),g(Q)]$, given by
$$K(x,y) = \frac{1}{\sqrt{2\pi}} \frac{g(x) - g(y)}{x-y} \widehat {df}(y-x)$$ has the form
\begin{equation}\label{finiterankformula}
K(x,y) = \sum_{j=1}^N \phi_j(x)\overline{\phi_j}(y)
\end{equation}
where $\{\phi_1, \phi_2, \dots, \phi_N\}$ is a linearly independent set of $L^2(\R)$ functions. Let $x_0 \in \R$ and $I$ a small open interval centered at $x_0$.  Denote the complement of $I$ as $J$ and the restriction of $\phi_i$ to $J$ as $\phi_{i,J}$.  The $N \times N$ matrix with elements $(\phi_{i,J},\phi_{k,J})$ is continuous in $|I|$ so we can choose a small such interval $I$ such that the determinant of this matrix is non-zero and thus the vectors $\phi_{1,J}, \phi_{2,J}, \dots, \phi_{N,J}$ are linearly independent.  We can thus find $N$ functions $\psi_i \in C_0^{\infty}(J^{\mathrm o})$ so that the matrix with matrix elements $(\phi_{i,J}, \psi_k) = (\phi_i,\psi_k)$ has non-zero determinant.  
If we define $\gamma_j(x)$ for all $x \ \in \R$ by
$$\gamma_j(x) = \sum_i (\phi_i,\psi_j) \phi_i(x),$$
the $\phi_i$ are thus linear combinations of the $\gamma_j$ and at least for $x$ outside the support of the $\psi_j$ we have 
$$\gamma_j(x) = \int K(x,y)\psi_j(y) dy .$$  
It follows from the explicit form of $K$ that the $\phi_i$ can be taken continuous in a neighborhood of $\overline I$.  Since we can cover $\R$ with a countable number of such intervals we can assume all the $\phi_i$ are continuous functions.  In addition since $\widehat {df}(u)$ is continuous and $\widehat {df}(0) = [f]/\sqrt{2\pi}$ is non-zero, we can see from the explicit form of $K$ that $g$ is $C^1$ with 
$$ g'(x) = 2\pi \sum_j|\phi_j(x)|^2 /[f].$$
If $\phi_j(x) = 0$ for all $j$ this contradicts the fact that $g$ is strictly increasing and that $\widehat {df}(y-x)$ is non-zero for $|y-x|$ small.  Thus $g'(x)$ is positive for all $x$.

 Since $\mathcal{F}f(P)\mathcal{F}^{-1} = f(Q)$ and $\mathcal{F}g(Q)\mathcal{F}^{-1} = g(-P)$, $i[-g(-P), f(Q)]$ has the integral kernel 

$$\tilde{K}(\xi,\eta) =\sum_j \widehat {\phi}_j(\xi)\overline{\widehat {\phi}_j}(\eta). $$
It follows that what we have proved for $g$ is also true for $f$.  In particular $f \in C^1(\R)$ with 
\begin{equation}\label{fprime}
f'(\xi) =  2\pi \sum_j|\widehat{\phi_j}(\xi)|^2/[g].
\end{equation}
\end{proof}
We can actually prove much more:

\begin{thm}
Suppose the commutator $C = i[f(P),g(Q)]$ is positive, non-zero, and finite rank. Then the functions $f'$ and $g'$  are in $\mathcal{S}(\R)$.
\end{thm}

\begin{proof}
We write $$K(x,y) = \sum_{j=1}^N \phi_j(x)\overline{\phi_j}(y)$$ where the functions $\phi_j$ are in $L^2(\R)$ and linearly independent.  We have shown that these functions can be assumed continuous.  Let $v(y) = <\phi_1(y), \cdots, \phi_N(y)> \ \in \C^N$.  The span of $\{v(y): y \in \R\}$ is $\C^N$ for otherwise there is a non-zero vector $c \in \C^N$ perpendicular to this span. This implies $\sum \overline c_j \phi_j(y) = 0$ contradicting the fact that $C$ has rank $N$. Choose $N$ points $y_j$ so that $v(y_1), \cdots, v(y_N)$ are linearly independent.  Let $\zeta \in C_0^{\infty}(-1,1)$ be a real non-negative function with $\int \zeta(y)dy = 1$.  Let $\zeta_t(y) = t^{-1}\zeta(y/t)$ and $\psi_j (y)= \zeta_t(y_j-y)$.  Then $v_j := \zeta_t* v(y_j) = \int \psi_j(y) v(y) dy =  \\ <\int \psi_j(y) \phi_1(y)dy,\cdots, \int \psi_j(y)\phi_N(y)dy>$.  We choose $t>0$ small enough so that the $v_j$ span $\C^N$. Let 
$$\gamma_j(x) = \sum_k \phi_k(x) (\phi_k, \psi_j) = (v_j, v(x)).$$  The $\phi_j(x)$ are linear combinations of the $\gamma_j(x)$.  In fact a standard linear algebra calculation gives
$$v(x) = \sum_{j,l} A^{-1}_{j l} \gamma_l(x) v_j, $$ where $A_{i j} = (v_i,v_j) $

We will show that $||D^m\phi_j||_2 + ||D^m\widehat{\phi_j}||_2 < \infty$ for all $m$..  We proceed to show $D^m \phi_j \in L^2(\R) \cap L^{\infty}(\R)$ by induction. For $x \notin I: = \cup_j (y_j -t, y_j + t) $ we have 
$$\gamma_j(x) = \int_{\R}\frac{g(y) - g(x)}{y-x} \widehat{f'}(y-x) \psi_j(y)dy$$ or changing variables

$$\gamma_j(x) = \int_{\R}\frac{g(x+u) - g(x)}{u} \widehat{f'}(u) \psi_j(x+u) du.$$
We note that for all $i$, $\psi_i= 0 $ in a neighborhood of $I^c$ so that \\ $|\gamma_j(x)| \le c \int |\psi_j(y)|/|x-y|dy$ and thus $\phi_j(x)$ and thus $g'(x)$ are bounded in a neighborhood of $I^c$.  Then 

$$\gamma_j'(x) = \int_{\R}[\frac{g'(x+u) - g'(x)}{u}\psi_j(x+u) + \frac{g(x+u) - g(x)}{u}\psi_j'(x+u) ] \widehat{f'}(u) du$$  or

$$|\gamma_j'(x)| \le c\int \frac{|g'(y) - g'(x)|}{|y-x|}|\psi_j(y)|dy+ c\int \frac{|g(y) - g(x)|}{|y-x|}|\psi'_j(y)|dy$$

$$\le C(|x|+ 1)^{-1}$$ for $x$ in a neighborhood of $I^c$.   
It follows that $\phi_j'\in L^2(I^c)\cap L^{\infty}(I^c)$.

We now choose a different set of $y_j$, call them $\tilde y_j$,  and replace $I$ with \\ $J = \cup_j (\tilde y_j -t, \tilde y_j + t) $.  If for example $0 < |\tilde y_j - y_j|$ is small enough, the $v(\tilde y_j)$ will be linearly independent and if $t$ is small enough the $v_j$ will also be linearly independent and $\overline I \cap \overline J = \emptyset$.

Using the same technique we find $\phi_j'\in L^2(J^c)\cap L^{\infty}(J^c)$. Thus we finally conclude the first step:  The $\phi_j$ are differentiable and  $\phi_j, \phi'_j \in  L^2(\R)\cap L^{\infty}(\R)$.  We now suppose $D^n\phi_j \in  L^2(\R)\cap L^{\infty}(\R)$ for $1\le n\le m$ and $1 \le j \le N$.  Using $g'(x) = 2\pi(\sum_j|\phi_j(x)|^2)/[f] $, $c = 2\pi/[f]$, we calculate  $$D^{n+1}g(x) =c \sum_{j,k} {n \choose k} [D^k\phi_j(x) D^{n-k}\overline \phi_j(x)].$$  This gives $D^{n+1}g \in L^1(\R)\cap L^{\infty}(\R)$ for $1 \le n \le m$.  Again using
$$\gamma_j(x) = \int_{\R}\frac{g(x+u) - g(x)}{u} \widehat{f'}(u) \psi_j(x+u) du,$$ we see that $\gamma_j$ is $m + 1 $ times continuously differentiable in a neighborhood of $I^c$ with 
$$D^{m+1}\gamma_j(x)= \sum_k{m+1 \choose k}\int \frac{D^kg(x+u) - D^kg(x)}{u} \widehat{f'}(u)D^{m+1-k} \psi_j(x+u) du$$ which gives
$$|D^{m+1}\gamma_j(x)| \le C\sum_k \int \frac{|D^kg(x+u) - D^kg(x)|}{|u|} |D^{m+1-k} \psi_j(x+u) |du$$
$$= C\sum_k \int \frac{|D^kg(y) - D^kg(x)|}{|y-x|} |D^{m+1-k} \psi_j(y) |dy \le C'(|x| + 1)^{-1}$$ in a neighborhood of  $I^c$.  So $D^{m+1}\gamma_j$ and hence the $D^{m+1}\phi_k$ are in $L^2(I^c)\cap L^{\infty}(I^c)$.  Using $J$ instead of $I$ we complete the induction step to learn  $D^{m+1}\phi_j \in L^2(\R)\cap L^{\infty}(\R)$.

As before $\mathcal{F}f(P)\mathcal{F}^{-1} = f(Q)$ and $\mathcal{F}g(Q)\mathcal{F}^{-1} = g(-P)$, so $i[-g(-P), f(Q)]$ has the integral kernel 

$$\tilde{K}(\xi,\eta) =\sum_j \widehat {\phi}_j(\xi)\overline{\widehat {\phi}_j}(\eta). $$
It follows that what we have proved for $\phi_j$ is also true for $\widehat \phi_j$.  This completes the proof.

\end{proof}

In case the commutator $C$ has finite rank with kernel as in (\ref{finiterankformula}) an important part of the kernel can be expressed entirely in terms of the $\phi_j$'s.  In particular we have

 $$\widehat {f'}(y-x) = (2\pi/[g])\sum_j \int \phi_j(x+u) \overline {\phi_j}(y+u)du.$$  There is another way of seeing this but let us calculate directly.  Using the formula (\ref{fprime}) for $f'$ we should calculate
$$ \int e^{-i(y-x)\xi}|\widehat \phi(\xi)|^2d\xi /\sqrt{2\pi} = \int \int\int e^{i(x-y)\xi}\phi(x_1)e^{-i\xi x_1} e^{i\xi x_2} \overline \phi(x_2)dx_1dx_2d\xi(2\pi)^{-3/2}$$
$$= \int \phi(x-y +x_2)\overline \phi(x_2)/\sqrt{2\pi} $$
$$= \int \phi(x+u) \overline {\phi}(y+u)du/\sqrt{2\pi} .$$

\section{Operator monotone functions}
The following proposition gives another connection between operator monotone functions and positive commutators. 
\begin{proposition}
Suppose $f$ and $g$ are bounded measurable functions with \\
 $i[f(P), g(Q)] \ge 0$.  Suppose $F$ is an operator monotone function on an open interval containing the range of $f$.  Then $i[F(f(P)),g(Q)] \ge 0$. 
\end{proposition}
\begin{proof}
Looking at the derivative, we see that if $t > 0,  e^{-itg(Q)}f(P)e^{itg(Q)} \ge f(P)$.  It follows that $e^{-itg(Q)}F(f(P))e^{itg(Q)} = F(e^{-itg(Q)}f(P)e^{itg(Q)}) \ge F(f(P))$.  The result now follows from the fact that the derivative of the left side of the inequality at $t=0$ is thus positive.

\end{proof}

\section{Two additional representations of the commutator}

The integral kernel of the commutator $i[f(P),g(Q)]$ given in (\ref{kernel2}) is clearly not very symmetric in the functions $f$ and $g$ although there is much symmetry in the operator (see (\ref{f&g}) for example). In this section we remedy this by giving two symmetrical representations.   \\
 
The integral kernel of our first symmetrical representation of $i[f(P),g(Q)] $ which is on $L^2 (\mathbb R ^2)$ appears in the next lemma.

\begin{lemma}

Suppose $i[f(P),g(Q)] = C$, where $f$ and $g$ are real measurable bounded functions and $C$ is positive and non-zero.  Then a representation of the commutator on $L^2(\mathbb R^2)$ has integral kernel 
$$ K(x_1, x_2 ; y_1, y_2) = \widehat{df}(\sqrt 2(y_1-x_1)) \widehat{dg}(\sqrt 2(y_2-x_2)) e^{i(x_1y_2-x_2y_1)} \frac{\sin(y_1-x_1)(y_2-x_2)}{(y_1-x_1)(y_2-x_2)}/\pi$$
\end{lemma}

\begin{proof}
We define a different representation of the canonical commutation relations.  We set
$ p_j = -i\partial_{x_j}$ and 

\begin{align*}
p = (p_1-x_2)/\sqrt 2,  x =&(x_1 + p_2)/\sqrt 2  \\
a = (x+ ip)/\sqrt 2\\
p' = (p_2 - x_1)/\sqrt 2,  x' =& (p_1 + x_2)/\sqrt 2\\
a' = (x' + i p')/\sqrt 2
\end{align*}
We have 
\begin{align*}
[a,a^*] = &[a',a'^*] = 1\\
[a,a'] = &[a,a'^*] =0
\end{align*}
For $f,g \in L^1(\R)$ we have

 \begin{align*}
f(p)\psi(x_1, x_2) =&\pi^{-1/2} \int \hat{f}(\sqrt 2(y_1-x_1))e^{ix_2(x_1-y_1)}\psi(y_1,x_2) dy_1\\
g(x)\psi(x_1,x_2) = & \pi^{-1/2} \int \hat{g}(\sqrt 2(y_2-x_2))e^{-ix_1(x_2-y_2)}\psi(x_1,y_2) dy_2. 
\end{align*}
Thus
\begin{align*}
f(p)g(x)\psi(x_1,x_2)=
\end{align*}

\begin{align*}
&\pi^{-1}\int \hat{f}(\sqrt 2(y_1-x_1))\hat{g}(\sqrt 2(y_2-x_2))e^{ix_2(x_1-y_1) - iy_1(x_2-y_2)}\psi(y_1,y_2)dy_1dy_2\\
=&\pi^{-1}\int \hat{f}(\sqrt 2(y_1-x_1))\hat{g}(\sqrt 2(y_2-x_2))e^{i(x_2-y_2)(x_1-y_1) +i( x_1y_2 - x_2y_1)}\psi(y_1,y_2)dy_1dy_2
\end{align*}
and
\begin{align*}
g(x)f(p)\psi(x_1,x_2)=
\end{align*}

\begin{align*}
&\pi^{-1}\int \hat{f}(\sqrt 2(y_1-x_1))\hat{g}(\sqrt 2(y_2-x_2))e^{-ix_1(x_2-y_2) + iy_2(x_1-y_1)}\psi(y_1,y_2)dy_1dy_2\\
=&\pi^{-1}\int \hat{f}(\sqrt 2(y_1-x_1))\hat{g}(\sqrt 2(y_2-x_2))e^{-i(x_2-y_2)(x_1-y_1) +i( x_1y_2 - x_2y_1)}\psi(y_1,y_2)dy_1dy_2.
\end{align*}
It follows that 
$$i[f(p),g(x)] = \int K(x_1,x_2;y_1,y_2) \psi(y_1,y_2) dy_1dy_2$$
where 
\begin{align*}
&K(x_1,x_2;y_1,y_2) \\
&= -2\pi^{-1} \hat{f}(\sqrt 2(y_1-x_1))\hat{g}(\sqrt 2(y_2-x_2))\sin(x_1-y_1)(x_2 - y_2)e^{i( x_1y_2 - x_2y_1)}\\
&=\pi^{-1} \widehat{df}(\sqrt 2(y_1-x_1))\widehat{dg}(\sqrt 2(y_2-x_2))\frac{\sin(x_1-y_1)(x_2 - y_2)}{(x_1-y_1)(x_2 - y_2)}e^{i( x_1y_2 - x_2y_1)}
\end{align*}
If $f,g$ are increasing and bounded it is not hard to make sense of these manipulations using distributions.  This proves the result.

\end{proof}

\begin{remark}
Note that if $\Omega$ is the vacuum vector ($a\Omega = a'\Omega = 0$), then $L^2(\mathbb R^2) \simeq \mathcal F_1\otimes \mathcal F_2$ where $\mathcal F_j$ are the Fock spaces which are the spans of $\{(a^*)^n\Omega: n = 0, 1, 2,...\}$ and $\{(a'^*)^n\Omega: n = 0, 1, 2,...\}$ and $ i[f(P),g(Q)] = k \otimes I$ for some operator $k$ acting in  $\mathcal F_1$.
\end{remark}

Another representation on $L^2(\R)$:

\begin{lemma}
$i[f(P),g(Q)] = (2\pi)^{-1}\int e^{i(\xi Q + u P)} \widehat{df}(u) \widehat{dg}(\xi) \frac{\sin(u\xi/2)}{u\xi/2}du d\xi$
\end{lemma}

\begin{proof}

We calculate
\begin{align*}
g(\infty) - g(x) = &\int _0^\infty dg(t+x)=\lim_{\epsilon \downarrow
 0}\int _0^{\infty} e^{-\epsilon t} dg(t+x)\\
=& \lim_{\epsilon \downarrow 0}(2\pi)^{-1/2}\int _0^{\infty}\int  e^{-\epsilon t} e^{i(t+ x)\xi} \widehat {dg}(\xi) d\xi dt\\
=&  \lim_{\epsilon \downarrow 0}(2\pi)^{-1/2}\int e^{ix\xi} \frac{\widehat {dg}(\xi)}{\epsilon - i\xi} d\xi
\end{align*}
and
\begin{align*}
[f(P), e^{iQ\xi}]= &e^{iQ \xi}(f(P+ \xi) - f(P)) \\
 =&e^{ix\xi} \int_0^1 f'(P+s\xi)\xi ds = (2\pi)^{-1/2}\xi \int_0^1 \int e^{iQ\xi}\e^{i(P+s\xi)u}
 \widehat {df}(u)duds
\end{align*}
and thus

\begin{align*}
i[f(P),g(Q)] = & \lim_{\epsilon \downarrow 0}-i (2\pi)^{-1}\int \int_0^1\int\frac{\xi}{\epsilon - i \xi} e^{iQ\xi}e^{i(P+s\xi)u}\widehat {df}(u)\widehat {dg}(\xi)dudsd\xi\\
=&(2\pi)^{-1}\int \int_0^1 \int e^{iQ\xi}e^{i(P+s\xi)u}\widehat {df}(u)\widehat {dg}(\xi)dudsd\xi\\
=&(2\pi)^{-1}\int\int_0^1 \int e^{i(Q\xi +Pu)}e^{i(s-1/2)\xi u}\widehat {df}(u)\widehat {dg}(\xi)dudsd\xi\\
=&\pi^{-1}\int \int e^{i(Q\xi +Pu)}\frac{\sin(\xi u/2)}{\xi u}\widehat {df}(u)\widehat {dg}(\xi)du d\xi.
\end{align*}
Since we do not know apriori that $f$ is differentiable we should first mollify our bounded $f$ and then take a limit at the end of the calculation.

\end{proof}

As an application of this representation we calculate expectations of the commutator in coherent states:
With $a = (Q+iP)/\sqrt 2$ and $\Omega = \pi^{-1/4}e^{-x^2/2}$ we have $a\Omega = 0$.  Let $\psi (z) = e^{za^*}\Omega$.  Then $a\psi(z) = z\psi(z)$.
Let $\zeta = (\xi + i u)/\sqrt 2$.  We calculate
$$e^{i(\xi Q + u P)} = e^{i(\zeta a^* + \bar{\zeta} a)} = e^{|\zeta|^2/2}  e^{i\bar{\zeta} a} e^{i\zeta a^*}$$
$$(\psi(w),\psi(z)) = e^{\bar{w} z}.$$
Thus 

\begin{align}
(\psi(w),i[f(P),g(Q)]\psi(z))&= \pi^{-1} \int e^{|\zeta|^2/2} e^{(\overline{w -i \zeta}) (z +i \zeta)}\frac{\sin(\xi u/2)}{\xi u}\widehat {df}(u)\widehat {dg}(\xi)du d\xi  \nonumber \\
&=  \pi^{-1}e^{\bar w z}\int e^{-|\zeta|^2/2 + i(\zeta \bar w + \bar \zeta z)} \frac{\sin(\xi u/2)}{\xi u}\widehat {df}(u)\widehat {dg}(\xi)du d\xi
\end{align}
Putting $z=w$ we obtain 

\begin{align}
\int e^{-(\xi^2 + u^2)/4} e^{ i(\xi x + u y)}\frac{\sin(\xi u/2)}{\xi u}\widehat {df}(u)\widehat {dg}(\xi)du d\xi \ge 0
\end{align}
for all $z =( x+ iy)/\sqrt 2$.

\section{Some results that follow from $2\times2$ positivity}

We assume that $i[F(P),G(Q)] = C \ge 0, C\ne 0$.  It then follows that $F$ and $G$ can be taken strictly increasing, continuous, with inverses which are absolutely continuous.  We now look more carefully at the condition that the kernel $$H_{xy} =  \frac{G(x) - G(y)}{x-y} \widehat{dF}(y-x)$$ gives a positive semidefinite 2 by 2  matrix for any pair $(x,y)$. In this section only we drop the factor of $\sqrt{2\pi}$ and use $\widehat{dF}(x) = \int e^{-ix\xi}dF(\xi)$.  Let us restrict to $F$ with $\int dF = 1$.    An important parameter will make its appearance, namely $$\sigma^2 = \inf\{\int \xi^2 df(\xi) -  (\int \xi df(\xi))^2:  i[f(P),G(Q)] = C \ge 0, C\ne 0, \int df = 1\}$$
which we believe must be related to the Kato class of $G$ (see some discussion below). We assume $\sigma^2 < \infty$.  Let $\hat a = \sqrt{3} \sigma$.

\begin{thm}\label{2x2}
Suppose  $i[F(P),G(Q)] = C \ge 0, C\ne 0$ where $F$ and $G$ are real bounded measurable functions.  Then both are monotone increasing or monotone decreasing and continuous.  Suppose there exists an increasing $f$ with $i[f(P),G(Q)] \ge 0$ and with $\int \xi^2 df(\xi) < \infty, \int df = 1$. Then in fact $G$ is $C^1$ with a Lipschitz derivative and thus $G'$ is absolutely continuous.  $G$ satisfies the following estimates:
\begin{align} 
& G'(x_0) e^{-2\hat a|(x-x_0)|} \le G'(x)  \le G'(x_0) e^{2\hat a|(x-x_0)|}\\
&|G''(x)| \le 2\hat a G'(x), \  \text{a.e. and both} \\  
& G'(x) \le 2\hat a|G(\pm \infty) - G(x)|.
\end{align}
If $G$ is odd then
\begin{align}
G'(x) \ge \frac{G'(0)}{(\cosh(\hat a x))^2}.
\end{align}
\end{thm}

This inequality is an equality in the rank one case \cite{TK} and if K is true the rank one case is the case of minimal variance (see Eq. \ref{variance}).

The proof will require some preliminary work.  Let $g_t = \phi_t * G$ where $\phi$ is a non-negative smooth function of compact support whose integral is $1$ and $\phi_t(x) = t^{-1} \phi(t^{-1}x)$  We take $f$ to satisfy $ i[f(P),G(Q)] = C \ge 0, C\ne 0$ and $\int df = 1$ and additionally  $\int \xi^2 df(\xi) < \infty$. The idea is to get uniform bounds on the derivatives of $g_t$ and then take limits to learn about $G$.  In the following we drop the subscript $t$.

\begin{lemma}
The function $g'$ is positive.  The function  $\psi = (g')^{-1/2}$  satisfies
\begin{equation} \label{B}
-\psi'' + \hat a^2 \psi \ge 0
\end{equation}
\end{lemma}

\begin{proof}
Let $ y=x+h$ with $h$ small and nonzero.  Then

\begin{align}
\frac{g(x) - g(y)}{x-y} &= (g(x+h) - g(x))/h = g'(x) + hg''(x)/2 + h^2g'''(x)/6 + o(h^2)\\
\widehat{df}(y-x) &= \int e^{-ih\xi} df(\xi) = e^{-ih\langle \xi\rangle } \int e^{-ih(\xi - \langle \xi\rangle )} df(\xi)  \nonumber \\
& = e^{-ih\langle \xi\rangle }(1 -\frac{h^2}{2}(\langle\xi^2\rangle - \langle\xi\rangle^2)+ o(h^2))
\end{align}
where we have written $\int\xi^n df(\xi) = \langle\xi^n\rangle$.
We have 

\begin{align}
&H_{xx}H_{yy} - |H_{xy}|^2 =   \nonumber \\
&g'(x)g'(y) - (\frac{g(x) - g(y)}{x-y})^2 |\widehat{df}(y-x)|^2 \nonumber \\ 
& = g'(x)(g'(x) +hg''(x) + h^2g'''(x)/2) - (g'(x) + hg''(x)/2 + h^2g'''(x)/6)^2 + \nonumber \\
& (g'(x) + hg''(x)/2 + h^2g'''(x)/6)^2 (1-|\widehat{df}(y-x)|^2) + o(h^2) \nonumber \\
& = h^2g'g'''/6 - h^2(g'')^2/4 +(g')^2(1 - [1-h^2(\langle \xi^2\rangle - \langle\xi\rangle^2)]) + o(h^2) \nonumber \\
& =\frac{ h^2}{6}[g'g''' - 3(g'')^2/2 +6(g')^2 (\langle \xi^2\rangle - \langle\xi\rangle^2)] + o(h^2)
\end{align}
We thus find 
\begin{equation} 
g'(x)g'''(x) - 3(g''(x))^2/2 + 6g'(x)^2(\langle \xi^2\rangle - \langle\xi\rangle^2) \ge 0
\end{equation}
Since this is true for all monotone $f$ with $\int df = 1$, $\int \xi^2 df < \infty$ and  $i[f(P),G(Q)]$  positive semidefinite, we take the infimum of the variance of $df$ over all such $f$ to find 
\begin{equation} \label{maineqn}
g'(x)g'''(x) - 3(g''(x))^2/2 + 6\sigma^2 g'(x)^2 \ge 0
\end{equation}
Since $g'(x) > 0$ (see the formula for the determinant above and use the fact that $g$ is strictly increasing) we can introduce $\psi = (g')^{-1/2}$.  Then recalling $\hat a^2 = 3 \sigma^2$, (\ref{maineqn}) becomes
\begin{equation}
-\psi'' + \hat a^2 \psi \ge 0
\end{equation}
\end{proof}
To see why we have used the abbreviation $\hat a^2 = 3 \sigma^2$, suppose we assume K and take $g$ in its smallest Kato class, $K_a$.  Let us find $\sigma^2$.  We must have $f \in K_{\hat a}$, viz. $f(\xi) = \tanh a \xi * \mu/2 + c$ for some positive measure $\mu$ and where $\hat a = \pi/2a$. Then if $\mu$ is a probability measure, $\int f'(\xi)d\xi = 1$.  We compute 
\begin{equation}\label{variance}
\langle \xi^2 \rangle - \langle \xi\rangle^2= a^{-2}[ \frac{\int \xi^2\cosh^{-2}(\xi) d\xi }{\int \cosh^{-2}(\xi) d\xi }]+ \int t^2 d\mu - (\int t d\mu)^2
\end{equation}
This is minimized by taking $\mu$ a point measure.  We compute with the help of Gradshteyn and Ryzhik 3527.3, $\int _{-\infty}^{\infty} x^2 \cosh^{-2}(x) dx = \pi^2/6$ and thus
\begin{equation}
\sigma^2 = \pi^2/12a^2 = \hat a^2/3
\end{equation}
Continuing with the proof of Theorem \ref{2x2}  we set  $u =\psi'/\psi$ and obtain the Ricatti type inequality
$$u' \le \hat a^2 - u^2.$$

\begin{lemma} \label{l-lemma}
$|\psi'(x)/\psi(x)| < \hat a$
\end{lemma}

\begin{proof}
If $u(x_0) > \hat a$ then as long as $u(x) \ge \hat a, \hat a^2 - u(x)^2 \le -(u(x) - \hat a)^2$ so that if $w = u - \hat a, w' \le -w^2$.  Thus $w$ tends to $ +\infty$ for some $x < x_0$, a contradiction.  If $u(x_0) < -\hat a$ then as long as $u(x) \le -\hat a, u' \le -(u+ b)^2$.  This leads to $u(x)$ tending to $+ \infty$ for some $x > x_0$, another contradiction.  Suppose  $u(x_0) = \hat{a}$. We know that $|u(x)| \le \hat{a}$ for all $x$.  If $\epsilon > 0 $ let $\hat{a}' = \hat{a} + \epsilon$.  We know that $-\psi'' + \hat{a}'^2 \psi \ge 0$. A computation shows that $\frac{\hat{a}' +u}{\hat{a}' -u} e^{-2\hat{a}' x} $ is monotone decreasing.  Thus if $x < x_0$,  $\frac{\hat{a}' +u(x) }{\hat{a}' -u(x)}e^{-2\hat{a}' x} \ge  \frac{\hat{a}' +\hat{a} }{\epsilon}e^{-2\hat{a} x_0}$ .  If we take $\epsilon \to 0$ we find $u(x) = \hat{a}$.  Integrating we find $\psi(x) = ce^{\hat{a} x}$ for $x < x_0$ or $g'(x) = e^{-2\hat{a} x}/c^2$.  This contradicts the boundedness of $g$.  Similarly if $u(x_0) = -\hat{a}$, $u(x) = - \hat{a}$ for all $x> x_0$ and this also contradicts the boundedness of $g$.  So we have shown $-\hat{a} < \psi'/\psi < \hat{a}$ for all $x$.

\end{proof}

Let $m(t) = \frac{\hat a + u}{\hat a-u},p(t) = \frac{1}{2}\log m(t)$.  Then for any $x_0 \in \R$
\begin{lemma}
\begin{equation} \label{summarizingformula}
g'(x) = g'(x_0) e^{-2\hat a \int_{x_{0}}^x \tanh(p(t))dt}; \   p'(t) \le \hat a.
\end{equation}

\begin{align} \label{g-inequalities}
& g'(x_0) e^{-2\hat a|(x-x_0)|} \le g'(x)  \le g'(x_0) e^{2\hat a|(x-x_0)|}  \nonumber \\
&|g''(x)| \le 2\hat a g'(x), \  \text{and both} \nonumber \\
& g'(x) \le 2\hat a|g(\pm \infty) - g(x)|. \nonumber \\ 
\end{align}

\end{lemma}
\vspace{.1in}

\begin{proof}
The fact that $p'(t) \le \hat a$ is a computation which uses $u' \le \hat a^2 -u^2$.  We compute $u = \frac{m-1}{m+1}\hat a = -\frac{1}{2} (\log g')'$.  Integration gives (\ref{summarizingformula}).  
The first two lines of (\ref{g-inequalities}) follow directly from (\ref{summarizingformula}). To prove the last inequalities first note that that $\lim_{x \to \infty} g'(x) = 0$.  To see this suppose the contrary.  Then there is a sequence $x_n$ with $x_{n+1} > x_n + 1$ so that $g'(x_n)  \ge \delta > 0$.  But then from the first inequality of the lemma $g'(x) \ge \delta e^{-2\hat a}$ for $x \in [x_n, x_n + 1]$.  This contradicts the integrability of $g'$.  Similarly  $\lim_{x \to - \infty} g'(x) = 0$.  From the second line of (\ref{g-inequalities}),
\begin{align*}
&- g''(x) \le 2\hat a  g'(x)  \ \text{or} \\
& (g'(x) + 2 \hat a g(x))' \ge 0. 
\end{align*}
Integrating from $x$ to infinity gives one of the last inequalities.  The other follows in a similar way.  
\end{proof}
To make a connection with the conjecture K consider the case where $g'$ is even.  
\begin{corollary} \label{coshinequality}

\begin{align}
g'(x) \ge \frac{g'(0)}{(\cosh \hat a x)^2}
\end{align}

\end{corollary}

\begin{proof}
This follows directly from $p'(t) \le \hat a$ and the formula in (\ref{summarizingformula}).  This inequality is an equality in the rank one case \cite{TK} and if K is true the rank one case is the case of minimal variance (see Eq. \ref{variance}).
\end{proof}

\begin{proof}[Proof of Theorem \ref{2x2}]
We have $g =g_t= \phi_t* G$ satisfying $ 0 < g_t'(x) \le 2 \hat a[G], |g_t''(x)| \\ \le (2 \hat a)^2 [G]^2$.  By the Arzela - Ascoli theorem there is a sequence $g_n = g_{t_{n}}, t_n \downarrow 0$, such that $g_n \to G_1$ and $g'_n \to G_2$ uniformly on compact subsets of $\mathbb{R}$.  Clearly $G_1 = G$ and $g_n(x) - g_n(y) = \int_x^y g_n'(t) dt$ so that $G(x) - G(y) = \int_x^y G_2(t)dt.$  Since $G_2$ is continuous it follows that $G$ is $C^1$ and $g_n' \to G'$ uniformly on compacts.  We also have $|g_n'(x) - g'_n(y)| = |\int_x^y g_n''(t) dt| \le (2 \hat a)^2[G]^2|x-y|$ giving $|G'(x) - G'(y)| \le (2 \hat a)^2[G]^2|x-y|$.  Thus $G'$ is absolutely continuous.  We have $g_n'(x+h) - g_n'(x) = \int_x^{x+h}g_n''(u)du \le \int_x^{x+h} 2\hat a g_n'(u) du$.  Thus  $G'(x+h) - G'(x) \le \int_x^{x+h} 2\hat a G'(u) du$.  It follows that $G''(x) \le 2\hat a G'(x)$.  Similarly $-G''(x) \le 2\hat a G'(x)$.  We have used that $g_n(\pm \infty) = G(\pm \infty)$. 
We take $\phi$ to be even and then $g = g_t$ is odd.  Thus the last inequality follows simply from Corollary \ref{coshinequality}.

 \end{proof}

If we now saturate (\ref{B}) by making the inequality an equality we obtain $g'^{-1/2} = \psi = be^{jx} + ce^{-jx}$ where $j = \sqrt 3 \sigma$.  Taking $b$ and $c$ positive to insure that $g'$ does not blow up at $\pm \infty$ we can write $\psi = \hat{a} \cosh(j (x-t)) $ for some $\hat{a} > 0$ and $t \in \mathbb{R}$.  This gives $g(x) = \hat{a}^{-2} \tanh (j(x-t)) + c$ for some constant $c$.  We can check that indeed $j = \hat a$ so that $g \in K_a$.  This is to be expected from Theorem \ref{half} which in particular states that if $i[f(P), g(Q)] \ge 0 $ for all $f \in K_{\hat a}$, then $g \in K_a$.  But it does show that the inequality (\ref{B}) is in some sense sharp, so that cases of equality occur if the Kato conjecture is correct.

\section{An interesting formula}

Here is an interesting formula:  Suppose $g$ is a bounded real function with a bounded analytic continuation to a strip of width $> 2\lambda$.  Consider 
\begin{align} \label{formula}
&i[\tanh \lambda P, g(Q)] = -2i[(1+ e^{2\lambda P})^{-1}, g(Q)]  \nonumber\\
& = 2i(1+ e^{2\lambda P})^{-1}[e^{2\lambda P}, g(Q)] (1+ e^{2\lambda P})^{-1} \nonumber\\
& = 2i (1+ e^{2\lambda P})^{-1}\Big(e^{\lambda P} (e^{\lambda P} g(Q) e^{-\lambda P}) e^{\lambda P} - e^{\lambda P} (e^{-\lambda P} g(Q)e^{\lambda P}) e^{\lambda P} \Big)(1+ e^{2\lambda P})^{-1} \nonumber \\
& = (i/2) (\cosh \lambda P)^{-1} \big(e^{\lambda P} g(Q) e^{-\lambda P} - e^{-\lambda P} g(Q) e^{\lambda P}\big)(\cosh\lambda P)^{-1} \nonumber\\
&=  (i/2) (\cosh \lambda P)^{-1} (g(Q -i\lambda) - g(Q+ i \lambda)) (\cosh \lambda P)^{-1} \nonumber \\
&=  (\cosh \lambda P)^{-1} \text{Im}  g(Q+i\lambda) (\cosh \lambda P)^{-1} 
\end{align}
Given the problems which occurred with manipulating unbounded functions of $P$ in Section \ref{finiterankC}, we assure the reader that the computations above can be made rigorous.\\

Perhaps (\ref{formula}) makes Theorem \ref{half} intuitive.\\

\end{document}